\documentclass[11pt]{article}
\usepackage{amsmath,epsfig}
\newenvironment{proof}{\makebox[7ex][l]{\it Proof:\/}}{\hfill\/ \hfill\/ {\it
Q.E.D.} \vspace{0.5ex}\\}

\newtheorem{thm}{Theorem}

\begin{document}

\makeatletter

\title{\bf Robust and Gain-Scheduling ${\cal H}_2$ Control Techniques \\
for LFT Uncertain and Parameter-Dependent Systems}

\author{Fen Wu\thanks{%
E-mail: \texttt{fwu@ncsu.edu},
Phone: (919) 515-5268.  
} 
\vspace*{0.1in} \\
Dept. of Mechanical and Aerospace Engineering \\
North Carolina State University \\
Raleigh, NC 27695}

\makeatother

\date{February 5, 2026}

\maketitle

\begin{abstract}
This paper addresses the robust ${\cal H}_2$ synthesis problem for linear fractional transformation (LFT) systems subject to structured uncertainty (parameter) and white-noise disturbances. 
By introducing an intermediate matrix variable, we derive convex synthesis conditions in terms of linear matrix inequalities (LMIs) that enable both robust and  gain-scheduled controller design for parameter-dependent systems. 
The proposed framework preserves the classical white-noise and impulse-response interpretation of the ${\cal H}_2$ criterion while providing certified robustness guarantees,  thereby extending optimal ${\cal H}_2$ control beyond the linear time-invariant setting. 
Numerical and application examples demonstrate that the resulting robust ${\cal H}_2$ controllers achieve significantly reduced conservatism and improved disturbance rejection compared with conventional robust 
${\cal H}_\infty$-based designs. 
\end{abstract}



\section{Introduction}
\label{sec:intro}

Robust control methodologies such as ${\cal H}_\infty$ control and $\mu$-synthesis have become standard tools for the design of complex systems operating under uncertainty \cite{PacD93,BalDGPS91,ZhoDG96,DulP2000}. 
These approaches emphasize worst-case performance and guarantee robust stability and disturbance attenuation in the presence of structured or unstructured uncertainties. 
When an uncertain system ${\bf T}$ is robustly stable, performance is typically assessed by the ${\cal L}_2$-gain from disturbance signals to regulated outputs, yielding a conservative but reliable characterization of system behavior under adversarial conditions.

However, the strong emphasis on worst-case robustness inherent in ${\cal H}_\infty$-based design often leads to unnecessarily conservative performance, particularly in applications where disturbances are stochastic in nature rather than adversarial. When disturbances are better modeled as white noise processes, the ${\cal H}_2$ norm provides a more meaningful performance metric, as it quantifies the expected energy of the regulated output rather than the worst-case amplification \cite{ZhoDG96}. In contrast, the induced ${\cal L}_2$ norm underlying ${\cal H}_\infty$ control effectively measures the response to persistent sinusoidal disturbances, which can significantly overestimate performance degradation in many practical systems.

Despite its appealing interpretation and widespread use in linear time-invariant (LTI) systems, the extension of the ${\cal H}_2$ norm to uncertain and parameter-varying systems remains challenging. 
Early attempts to generalize ${\cal H}_2$ performance to uncertain systems were shown to be basis dependent and computationally demanding \cite{Sto93,ZhoGBD94}. 
Subsequent research has sought to address these limitations by developing alternative formulations of robust ${\cal H}_2$  performance analysis. 
For example, reference \cite{Fer97} provided analysis results of robust ${\cal H}_2$
performance for uncertain linear systems using multipliers. 
References \cite{Pag99,Pag99a} provided rigorous interpretations of robust ${\cal H}_2$ measures that preserve the white-noise rejection perspective.
These works also demonstrated that tight performance bounds can be obtained and that infinite-dimensional conditions may be reduced to tractable state-space characterizations.
A comprehensive comparison of these robust ${\cal H}_2$ measures and their theoretical properties is reported in \cite{PagF2000}.
Reference \cite{GohW97} developed a rigorous robust
${\cal H}_2$ performance analysis framework based on duality theory and basis-function representations, preserving the classical white-noise interpretation while explicitly accounting for structured uncertainty.
Their results yield tight, finite-dimensional state-space conditions for computing worst-case
${\cal H}_2$ performance bounds and form a theoretical foundation for later multiplier- and LMI-based robust
${\cal H}_2$ methods.

In parallel, gain-scheduling and linear parameter-varying (LPV) control frameworks have been extensively investigated as a means to address structured parameter variations in control systems. 
Quadratic and parameter-dependent ${\cal H}_2$ and mixed ${\cal H}_2/{\cal L}_\infty$ performance analysis and synthesis techniques have been proposed for LPV systems, including full-block multiplier and loop-shaping-based methods \cite{PreP2008}, static output-feedback designs with guaranteed ${\cal H}_2$ performance bounds \cite{ZhaL2018}, and gain-scheduled controllers accounting for inexact scheduling parameter measurements \cite{SatEP2010,DaaBG2008}. 
Relevant multi-objective control problems have been addressed in \cite{ZhoGBD94,SchGC2002}. 
Although substantial advances have been made, current approaches frequently depend on restrictive assumptions, incur conservatism arising from insufficient disturbance characterization, or handle robustness and gain scheduling in a loosely coupled manner. 

Motivated by these limitations and the later robust ${\cal H}_2$ problem formulation in \cite{Pag99,Pag99a}, this paper investigates robust and gain-scheduled ${\cal H}_2$ control from a unified perspective. 
The objective is to explicitly account for uncertainty while retaining the physical interpretability and performance advantages of the ${\cal H}_2$ criterion. 
In contrast to classical worst-case designs, the proposed framework seeks to balance robustness and average performance by embedding uncertainty descriptions directly into the ${\cal H}_2$ analysis and synthesis conditions. 

This paper develops a unified framework for robust and gain-scheduled ${\cal H}_2$ control of uncertain and parameter-dependent systems, addressing a fundamental gap between classical ${\cal H}_2$ optimal control and robust design methodologies. By explicitly incorporating uncertainty into the ${\cal H}_2$ performance formulation, the proposed approach preserves the physical interpretation of white-noise disturbance rejection while providing certified robustness guarantees.
For time-varying uncertainty and scheduling parameters, finite-dimensional synthesis conditions are derived that parallel robust and gain-scheduled ${\cal H}_\infty$ control techniques. 
Central to the proposed method is the introduction of an intermediate matrix variable and the use of ideas originally proposed in \cite{DeGB2002,EbiH2005}, which enable the derivation of convex synthesis conditions expressed in terms of tractable LMIs (or PLMIs).
The resulting framework unifies classical LTI ${\cal H}_\infty$ control, robust ${\cal }_2$ synthesis, and LPV gain-scheduled designs within a single theoretical setting. In contrast to worst-case induced-norm approaches, the proposed methodology provides performance guarantees that are better aligned with stochastic disturbance environments and impulse-response objectives. Several numerical and application examples are presented to demonstrate that the proposed robust ${\cal H}_2$ control techniques achieve reduced conservatism and improved disturbance rejection when compared with existing robust and gain-scheduled ${\cal H}_2$ and ${\cal H}_\infty$-based methods.  
Parts of this work were previously reported in a preliminary conference paper \cite{WuD2005}.  

{\bf Notations}. 
${\bf R}^p$ and ${\bf R}^{p \times q}$ denote real-valued vectors of dimension 
$p$ and real-valued matrices of dimension $p \times q$, respectively. 
The set of real symmetric $n \times n$ matrices is denoted by ${\bf S}^{n\times n}$, and ${\bf S}^{n\times n}_+$ denotes the subset of positive definite matrices.
A block-diagonal matrix with diagonal blocks $X_1, \ldots, X_p$ is denoted by ${\rm diag} \left\{ X_1, \ldots, X_p \right\}$.
For a matrix $X \in {\bf R}^{p \times p}$ with elements $x_{ij}$, ${\rm diag} \left\{
X \right\}$ denotes the vector $[x_{11}, \ldots, x_{pp}]^T$, and ${\rm tr}
\left\{X \right\} = \sum_{i=1}^p x_{ii}$ denotes its trace. 
For symmetric matrices $X, Y \in {\cal S}^{n \times n}$, the notation 
$X \leq Y$ ($X < Y$) that $Y - X$ is positive 
semidefinite (positive definite). 
An infinite sequence $x := \{ x_1, x_2, \ldots \}$, with $x_i \in {\bf R}^p$,
is said to belong to ${\ell}_2^{p}$ if $\sum_{i=1}^{\infty} x_i^T x_i < \infty$. 

The remainder of the paper is organized as follows. 
In Section 2, we introduce the system model, including the linear fractional transformation (LFT) representation of the uncertain plant, and summarize the assumptions on the uncertainty and system properties. Section 3 presents the robust ${\cal H}_2$ state-feedback control design, including the derivation of convex synthesis conditions and the corresponding controller construction. 
Section 4 extends the results to gain-scheduling 
${\cal H}_2$ output-feedback control for parameter-dependent systems, with detailed LMI-based synthesis conditions and controller recovery. 
Section 5 provides numerical examples to illustrate the effectiveness and advantages of the proposed methods. 
Finally, Section 6 concludes the paper and discusses potential directions for future research.
 
\section{Robust ${\cal H}_2$ Performance Analysis}
\label{sec:robh2ana}

Consider an uncertain discrete-time linear system described in linear fractional transformation (LFT) form
\begin{align}
\begin{bmatrix}
x(k+1) \\
q(k) \\
e(k)
\end{bmatrix} & = \begin{bmatrix}
A & B_0 & B_1 \\
C_0 & D_{00} & 0 \\
C_1 & D_{10} & 0
\end{bmatrix}
\begin{bmatrix}
x(k) \\
p(k) \\
d(k)
\end{bmatrix} \label{eqn:ulsys1} \\
p(k) & = \Delta(k) q(k) \label{eqn:ulsys2}
\end{align}
where $x(k), x(k+1) \in {\bf R}^n$, $p(k), q(k) \in {\bf R}^{n_p}$,
$e \in {\bf R}^{n_e}$, and $d \in {\bf R}^{n_d}$.
All state-space matrices are of compatible dimensions, and the nominal matrix $A$ is assumed to be stable. 
The uncertainty $\Delta$ time-varying and belongs to the structured set 
\begin{align*}
{\boldsymbol \Delta} & := \left\{ {\rm diag} \left\{\delta_1 I_{m_1},
\ldots, \delta_s I_{m_s}, \Delta_{s+1}, \ldots, \Delta_{s+f} \right\}: \delta_i: {\ell}_2 \rightarrow {\ell}_2, \right. \\
& \left. \hspace*{0.25in} \| \delta_i \| \leq 1, \ i = 1, \ldots, s, \ \Delta_{s+j}: {\ell}_2 \rightarrow
{\ell}_2^{r_j \times r_j}, \| \Delta_{s+j} \| \leq 1, \ j = 1, \ldots, f \right\},
\end{align*}
where the induced $\ell_2$ norm is used, and $\sum_{i=1}^s m_i + \sum_{j=1}^f r_j = n_p$. 
To reduce conservatism in robust stability and performance analysis, we introduce
the following set of scaling matrices: 
\begin{align*}
\mathcal{D} & = \left\{ {\rm diag} \left\{D_1, \cdots, D_s, d_{s+1} I_{r_1}, \cdots, d_{s+f} I_{r_f} \right\}:
\right. \\
& \left. \hspace*{0.25in} D_i \in {\bf S}^{m_i \times m_i}_+, \ i = 1, \cdots, s, \
d_{s+j} > 0, \ j = 1, \cdots, f \right\}
\end{align*}
By construction, the commutation property $D \Delta = \Delta D$ holds for any $\Delta \in
{\boldsymbol \Delta}$ and $D \in \mathcal{D}$.

Following \cite{Pag99}, defining the family of signal sets
\[
\mathcal{W}_{\eta} := \left\{ f \in \ell_2^{p}: \; \left\| \int_0^s f(\omega)
f^{\star}(\omega) \: \frac{d \omega}{2\pi} - \frac{s \| f \|^2}{2p\pi} I_p
\right\|_{\infty} < \eta \right\},
\]
where $\| \cdot \|_{\infty}$ denotes the maximum absolute value of the matrix elements over $\omega \in [0, 2\pi]$, $\| f \|$
denotes the ${\ell}_2^{p}$ norm of $f$, and $f(\omega)$ denotes the Fourier transform of $f$. 
For a (not necessarily linear or time-invariant) system
$\mathbf{T}: \ell_2^{p} \rightarrow \ell_2^{q}$, define 
\[
\| \mathbf{T} \|_{\mathcal{W}_{\eta}} := \sup_{f \in \mathcal{W}_\eta, \; \| f
\|=1} \| \mathbf{T} f \|.
\]
It was then shown in \cite{Pag99} that
an appropriate extension of the ${\cal H}_2$ norm to such systems is given by 
\begin{equation}
\| \mathbf{T} \|_2 := \sqrt{p} \; \lim_{\eta \rightarrow 0} \| \mathbf{T}
\|_{\mathcal{W}_{\eta}}.
\label{eqn:h2normpag}
\end{equation}
This definition captures the white-noise response characteristics of ${\bf T}$ and is adopted throughout this paper.

For systems with time-varying uncertainty, the robust ${\cal H}_2$ performance of \eqref{eqn:ulsys1}-\eqref{eqn:ulsys2} was analyzed in the frequency domain in \cite{Pag99}. 
Specifically, the existence of a scaling matrix $X \in {\cal D}$ and a frequency-dependent matrix function $Y(e^{j \omega})$ satisfying 
\begin{align}
M^\star(e^{j \omega}) \begin{bmatrix}
X & 0 \\
0 & I
\end{bmatrix} M(e^{j \omega}) - \begin{bmatrix}
X & 0 \\
0 & Y(e^{j \omega})
\end{bmatrix} & < 0 \\
\int_{0}^{2 \pi} {\rm tr} \left[ Y(e^{j \omega}) \right] \frac{d \omega}{2 \pi} & < \gamma^2
\end{align}
ensures a robust ${\cal H}_2$ performance bound $\gamma$, 
where
\[
M(e^{j \omega}) = \begin{bmatrix} D_{00} \\ D_{10} \end{bmatrix} +
\begin{bmatrix} C_{0} \\ C_{1} \end{bmatrix} (e^{j \omega} I - A)^{-1} B_0. 
\]
Following the approach in \cite{Pag99a}, these frequency-domain conditions can be converted into an equivalent state-space formulation. 
In particular, robust ${\cal H}_2$ performance is guaranteed 
if there exist positive-definite $P_-, P_+$ and $Q$, together with a scaling matrix $X \in {\cal D}$, such that
\begin{align}
\begin{bmatrix}
P_- & & \\
& X & \\
& & I
\end{bmatrix} - \begin{bmatrix}
A & B_0 \\
C_{0} & D_{00} \\
C_{1} & D_{10}
\end{bmatrix}
\begin{bmatrix}
P_- & \\
& X
\end{bmatrix}
\begin{bmatrix}
A^T & C_{0}^T & C_{1}^T \\
B_0^T & D_{00}^T & D_{10}^T
\end{bmatrix} & > 0 \label{eqn:h2ana1} \\
\begin{bmatrix}
P_+ & & \\
& X & \\
& & I
\end{bmatrix} - \begin{bmatrix}
A & B_0 \\
C_{0} & D_{00} \\
C_{1} & D_{10}
\end{bmatrix}
\begin{bmatrix}
P_+ & \\
& X
\end{bmatrix}
\begin{bmatrix}
A^T & C_{0}^T & C_{1}^T \\
B_0^T & D_{00}^T & D_{10}^T
\end{bmatrix} & > 0 \label{eqn:h2ana2} \\
\begin{bmatrix}
Q & B_1^T \\
B_1 & P_+ - P_-
\end{bmatrix} & > 0 \label{eqn:h2ana3} \\
{\rm tr} (Q) & < \gamma^2. \label{eqn:h2ana4}
\end{align}

These conditions provide a state-space characterization of
robust ${\cal H}_2$ performance that parallels classical
robust ${\cal H}_\infty$ results, thereby extending the
applicability of robust control theory to energy-based
performance measures.
Although the robust ${\cal H}_2$ analysis problem is now well
understood, there remains comparatively limited work on the synthesis of robust controllers based explicitly on the robust ${\cal H}_2$ criterion.


\section{Robust ${\cal H}_2$ State-Feedback Control}
\label{sec:robh2syn}

Consider the uncertain discrete-time linear system
\begin{align}
\begin{bmatrix}
x(k+1) \\
q(k) \\
e(k) \\
y(k)
\end{bmatrix} & = \begin{bmatrix}
A & B_0 & B_1 & B_2 \\
C_0 & D_{00} & 0 & D_{02} \\
C_1 & D_{10} & 0 & D_{12} \\
C_2 & D_{20} & D_{21} & 0
\end{bmatrix}
\begin{bmatrix}
x(k) \\
p(k) \\
d(k) \\
u(k)
\end{bmatrix} \\
p(k) & = \Delta(k) q(k), 
\end{align}
where $x(k) \in {\bf R}^n, u \in {\bf R}^{n_u}$, $y \in {\bf R}^{n_y}$, and other variables have compatible dimensions. We assume the following: 
\begin{description}
\item[(A1)]
$(A, B_2)$ is stabilizable and $(A, C_2)$ is detectable,
\item[(A2)]
$D_{01} = $ and $D_{11} = 0$.
\end{description}
In this section, we assume full-state availability ($y = x$) and seek a state-feedback controller 
\begin{align}
u(k) = F x(k)
\end{align}
that minimizes the robust ${\cal H}_2$ norm of the closed-loop system. The closed-loop dynamics are 
\begin{align*}
\begin{bmatrix}
x(k+1) \\
q(k) \\
e(k)
\end{bmatrix} 
& = \begin{bmatrix}
A + B_2 F & B_0 & B_1 \\
C_0 + D_{02} F & D_{00} & 0 \\
C_1 + D_{12} F & D_{10} & 0
\end{bmatrix}
\begin{bmatrix}
x(k) \\
p(k) \\
d(k)
\end{bmatrix} \\
p(k) & = \Delta(k) q(k). 
\end{align*}

To derive a tractable synthesis condition, we introduce an intermediate matrix $V$ as in \cite{DeGB2002,EbiH2005}, which replaces $P_-$ and $P_+$, providing common variables in the analysis inequalities and decoupling the state matrix $A$ from $P_-, P_+$. 
Using this transformation, the robust ${\cal H}_2$ state-feedback synthesis is given as follows. 

\begin{thm}
The closed-loop system is robustly stabilizable by a state-feedback controller and achieves a robust ${\cal H}_2$ norm less than $\gamma$ if there exist positive-definite matrices $P_-, P_+ \in
{\cal S}_+^{n \times n}$, $Q \in {\cal S}_+^{n_d \times n_d}$, a scaling
matrix $X \in {\cal D}$, and rectangular matrices $M \in {\bf R}^{n_u \times n}$
and $V \in {\bf R}^{n \times n}$ such that
\begin{align}
\begin{bmatrix}
V^T + V - P_+ & & (A V + B_2 M)^T & (C_0 V + D_{02} M)^T & (C_1 V + D_{12} M)^T \\
& X & X B_0^T & X D_{00}^T & X D_{12}^T \\
A V + B_2 M & B_0 X & P_+ & & \\
C_0 V + D_{02} M & D_{00} X & & X & \\
C_1 V + D_{12} M & D_{10} X & & & I
\end{bmatrix} & > 0 \label{eqn:h2LFTsf1} \\
\begin{bmatrix}
V^T + V - P_- & & (A V + B_2 M)^T & (C_0 V + D_{02} M)^T & (C_1 V + D_{12} M)^T \\
& X & X B_0^T & X D_{00}^T & X D_{12}^T \\
A V + B_2 M & B_0 X & P_- & & \\
C_0 V + D_{02} M & D_{00} X & & X & \\
C_1 V + D_{12} M & D_{10} X & & & I
\end{bmatrix} & > 0 \label{eqn:h2LFTsf2} \\
\begin{bmatrix}
Q & B_1^T \\
B_1 & P_+ - P_-
\end{bmatrix} & > 0 \label{eqn:h2LFTsf3} \\
{\rm tr}(Q) & < \gamma^2 \label{eqn:h2LFTsf4}
\end{align}
The resulting robust state-feedback controller is 
$u(k) = M V^{-1} x(k)$.
\end{thm}

\begin{proof}
Define $ F:= M V^{-1} $. Since $P_{+}, P_{-} > 0$, 
\begin{align*}
P_{+}^{-1} \geq V^{-T} ( V^T + V - P_{+} ) V^{-1} \\
P_{-}^{-1} \geq V^{-T} ( V^T + V - P_{-} ) V^{-1}
\end{align*}
for any nonsingular matrix $V$. 
Applying this inequality to the robust ${\cal H}_2$ analysis
conditions (\ref{eqn:h2ana1})-(\ref{eqn:h2ana4}), we have  
\begin{align}
\begin{bmatrix}
V^{-T} (V^T + V - P_+) V^{-1} & & (A + B_2 F)^T & (C_0 + D_{02} F)^T & (C_1 + D_{12} F)^T \\
& X^{-1} & B_0^T & D_{00}^T & D_{10}^T \\
A + B_2 F & B_0 & P_+ & & \\
C_0 + D_{02} F & D_{00} & & X & \\
C_1 + D_{12} F & D_{10} & & & I
\end{bmatrix} & > 0 \label{eqn:h2LFTsfp1} \\
\begin{bmatrix}
V^{-T} (V^T + V - P_-) V^{-1} & & (A + B_2 F)^T & (C_0 + D_{02} F)^T & (C_1 + D_{12} F)^T \\
& X^{-1} & B_0^T & D_{00}^T & D_{10}^T \\
A + B_2 F & B_0 & P_- & & \\
C_0 + D_{02} F & D_{00} & & X & \\
C_1 + D_{12} F & D_{10} & & & I
\end{bmatrix} & > 0 \label{eqn:h2LFTsfp2} \\
\begin{bmatrix}
Q & B_1^T \\
B_1 & P_+ - P_-
\end{bmatrix} & > 0 \label{eqn:h2LFTsfp3}\\
{\rm tr}(Q) & < \gamma^2. \label{eqn:h2LFTsfp4}
\end{align}
and performing the congruence transformation ${\rm diag} \left\{V^T, X, I, I \right\}$ 
leads to the LMI conditions above. 
\end{proof}

It is also possible to consider output-feedback controllers 
for robust ${\cal H}_2$ synthesis.
However, as in the robust ${\cal H}_\infty$ case, the resulting synthesis problem is non-convex, and will not be discussed here. 

\section{Gain-Scheduling ${\cal H}_2$ Output-Feedback Control}
\label{sec:gsh2syn}

If the uncertainty $\Delta(k)$ is measurable in real time, 
we can consider a parameter-dependent output-feedback
controller:
\begin{align}
\begin{bmatrix}
\dot{x}_k(k) \\
u(k) \\
q_k(k)
\end{bmatrix} & = \begin{bmatrix}
A_k & B_{k1} & B_{k0} \\
C_{k1} & 0 & D_{k10} \\
C_{k0} & 0 & D_{k00}
\end{bmatrix}
\begin{bmatrix}
x_k(k) \\
y(k) \\
p_k(k)
\end{bmatrix} \\
p_k(k) & = \Delta(k) q_k(k)
\end{align}
which mirrors the parameter dependence of the LFT plant and provides gain-scheduling control. 

The closed-loop system then takes the form 
\begin{align}
\left[ \begin{array}{c}
x(k+1) \\ \hline
q(k) \\
q_k(k) \\ \hline
e(k)
\end{array} \right] & = \begin{bmatrix}
A_{cl} & B_{0,cl} & B_{1,cl} \\
C_{0,cl} & D_{00,cl} & D_{01,cl} \\
C_{1,cl} & D_{10,cl} & D_{11,cl}
\end{bmatrix}
\left[ \begin{array}{c}
x(k) \\ \hline
p(k) \\
p_k(k) \\ \hline
d(k)
\end{array} \right] \\
\begin{bmatrix}
p(k) \\
p_k(k)
\end{bmatrix} & = \begin{bmatrix}
\Delta(k) & \\
& \Delta(k)
\end{bmatrix}
\begin{bmatrix}
q(k) \\
q_k(k)
\end{bmatrix}. 
\end{align}
The closed-loop matrices are defined as
\begin{align*}
\begin{bmatrix}
A_{cl} & B_{0,cl} & B_{1,cl} \\
C_{0,cl} & D_{00,cl} & D_{01,cl} \\
C_{1,cl} & D_{10,cl} & D_{11,cl}
\end{bmatrix} & = \left[ \begin{array}{cc|cc|c}
A & 0 & B_0 & 0 & B_1 \\
0 & 0 & 0 & 0 & 0 \\ \hline
C_0 & 0 & D_{00} & 0 & 0 \\
0 & 0 & 0 & 0 & 0 \\ \hline
C_1 & 0 & D_{10} & 0 & 0
\end{array} \right] + \left[ \begin{array}{ccc}
0 & B_2 & 0 \\
I & 0 & 0 \\ \hline
0 & D_{02} & 0 \\
0 & 0 & I \\ \hline
0 & D_{12} & 0
\end{array} \right] \\
& \hspace*{0.25in} \times \begin{bmatrix}
A_k & B_{k1} & B_{k0} \\
C_{k1} & 0 & D_{k10} \\
C_{k0} & 0 & D_{k00}
\end{bmatrix}
\left[ \begin{array}{cc|cc|c} 0 & I & 0 & 0 & 0 \\
C_2 & 0 & D_{20} & 0 & D_{21} \\
0 & 0 & 0 & I & 0
\end{array} \right]. 
\end{align*}
By construction, $D_{01,cl} = 0$ and $D_{00,cl} = 0$ hold as desired.

Similar to the state-feedback case, we introduce intermediate matrix variables $V$ (not necessarily symmetric) to reduce conservatism and derive a convex LMI synthesis condition. 

\begin{thm}
The closed-loop LFT system is exponentially stabilized by a gain-scheduling LFT output-feedback controller with robust ${\cal H}_2$ norm less than $\gamma$ if there exist 
positive-definite matrices $T_+, T_-
\in {\cal S}_+^{2 n \times 2 n}$, $Q \in {\cal S}_+^{n_d \times n_d}$,
scaling matrices $L, J \in {\cal D}$, and square matrices $R, S, U \in
{\bf R}^{n \times n}$, such that 
\begin{align}
\left[ \begin{array}{cc|ccc}
\begin{pmatrix} S+S^T & I+U^T \\ U+I & R^T+R \end{pmatrix} - T_+ & & G_{31}^T & G_{41}^T & G_{51}^T \\
& \begin{pmatrix} L & I \\ I & J \end{pmatrix} & G_{32}^T & G_{42}^T & G_{52}^T \\ \hline
G_{31} & G_{32} & T_+ & & \\
G_{41} & G_{42} & & \begin{pmatrix} L & I \\ I & J \end{pmatrix} & \\
G_{51} & G_{52} & & & I
\end{array} \right] & > 0 \label{eqn:h2LFTsyn1} \\
\left[ \begin{array}{cc|ccc}
\begin{pmatrix} S+S^T & I+U^T \\ U+I & R^T+R \end{pmatrix} - T_- & & G_{31}^T & G_{41}^T & G_{51}^T \\
& \begin{pmatrix} L & I \\ I & J \end{pmatrix} & G_{32}^T & G_{42}^T & G_{52}^T \\ \hline
G_{31} & G_{32} & T_- & & \\
G_{41} & G_{42} & & \begin{pmatrix} L & I \\ I & J \end{pmatrix} & \\
G_{51} & G_{52} & & & I
\end{array} \right] & > 0 \label{eqn:h2LFTsyn2} \\
\begin{bmatrix}
Q & H_{21}^T \\
H_{21} & T_+ - T_-
\end{bmatrix} & > 0 \label{eqn:h2LFTsyn3} \\
{\rm tr} (Q) & < \gamma^2 \label{eqn:h2LFTsyn4}
\end{align}
where $G_{3i}, G_{4i},G_{5i}$ and $H_{21}$ are defined as 
\begin{align*}
\begin{bmatrix}
G_{31} & G_{32} \\
G_{41} & G_{42} \\
G_{51} & G_{52}
\end{bmatrix} & = \left[ \begin{array}{cc|cc}
A S & A & B_0 & B_0 J \\
0 & R A & R B_0 & 0 \\ \hline
0 & L C_0 & L D_{00} & 0 \\
C_0 S & C_0 & D_{00} & D_{00} J \\ \hline C_1 S & C_1 & D_{10} &
D_{10} J
\end{array} \right] + \left[ \begin{array}{ccc}
0 & B_2 & 0 \\
I & 0 & 0 \\ \hline
0 & 0 & I \\
0 & D_{02} & 0 \\ \hline
0 & D_{12} & 0
\end{array} \right] \\
& \hspace*{0.25in} \times \begin{bmatrix}
\hat{A}_k & \hat{B}_{k1} & \hat{B}_{k0} \\
\hat{C}_{k1} & 0 & \hat{D}_{k10} \\
\hat{C}_{k0} & 0 & \hat{D}_{k00}
\end{bmatrix}
\left[ \begin{array}{cc|cc}
I & 0 & 0 & 0 \\
0 & C_2 & D_{20} & 0 \\
0 & 0 & 0 & I
\end{array} \right] \\
H_{21} & = \begin{bmatrix}
B_1 \\
R B_1
\end{bmatrix} + \begin{bmatrix}
0 & B_2 & 0 \\
I & 0 & 0
\end{bmatrix}
\begin{bmatrix}
\hat{A}_k & \hat{B}_{k1} & \hat{B}_{k0} \\
\hat{C}_{k1} & 0 & \hat{D}_{k10} \\
\hat{C}_{k0} & 0 & \hat{D}_{k00}
\end{bmatrix}
\begin{bmatrix}
0 \\
D_{21} \\
0
\end{bmatrix}. 
\end{align*}
The LFT output-feedback controller gains are then recovered by inverting the transformation 
\begin{align}
\begin{bmatrix}
A_k & B_{k1} & B_{k0} \\
C_{k1} & 0 & D_{k10} \\
C_{k0} & 0 & D_{k00}
\end{bmatrix} & = \begin{bmatrix}
M & R B_2 & 0 \\
0 & I & 0 \\
0 & L D_{02} & L_2
\end{bmatrix}^{-1}
\left( \begin{bmatrix}
\hat{A}_k & \hat{B}_{k1} & \hat{B}_{k0} \\
\hat{C}_{k1} & 0 & \hat{D}_{k10} \\
\hat{C}_{k0} & 0 & \hat{D}_{k00}
\end{bmatrix} - \begin{bmatrix}
R A S & 0 & R B_0 J \\
0 & 0 & 0 \\
L C_0 S & 0 & L D_{00} J
\end{bmatrix} \right) \nonumber \\
& \hspace*{0.25in} \times \begin{bmatrix}
N & 0 & 0 \\
C_2 S & I & D_{20} J \\
0 & 0 & J_2^T
\end{bmatrix}^{-1}
\end{align}
where $M N = U - R S$ and $L_2 J_2^T = I - L J$.
\end{thm}

\begin{proof}
Since for any nonsingular matrix $V$, we have 
\begin{align*}
P_{+,cl}^{-1} & \geq V^{-T} ( V + V^T - P_{+,cl} ) V^{-1} \\
P_{-,cl}^{-1} & \geq V^{-T} ( V + V^T - P_{-,cl} ) V^{-1}, 
\end{align*}
a sufficient condition that guarantees closed-loop stability and the robust ${\cal H}_2$
norm for the LFT system is
\begin{align}
\begin{bmatrix}
V^{-T} ( V + V^T - P_{+,cl} ) V^{-1} & & A_{cl}^T & C_{0,cl}^T & C_{1,cl}^T \\
& X^{-1} & B_{0,cl}^T & D_{00,cl}^T & D_{10,cl}^T \\
A_{cl} & B_{0,cl} & P_{+,cl} & & \\
C_{0,cl} & D_{00,cl} & & X & \\
C_{1,cl} & D_{10,cl} & & & I
\end{bmatrix} & > 0 \label{eqn:clptest1} \\
\begin{bmatrix}
V^{-T} ( V + V^T - P_{-,cl} ) V^{-1} & & A_{cl}^T & C_{0,cl}^T & C_{1,cl}^T \\
& X^{-1} & B_{0,cl}^T & D_{00,cl}^T & D_{10,cl}^T \\
A_{cl} & B_{0,cl} & P_{-,cl} & & \\
C_{0,cl} & D_{00,cl} & & X & \\
C_{1,cl} & D_{10,cl} & & & I
\end{bmatrix} & > 0 \label{eqn:clptest2} \\
\begin{bmatrix}
Q & B_{1,cl}^T \\
B_{1,cl} & P_{+,cl} - P_{-,cl}
\end{bmatrix} & > 0 \label{eqn:clptest3} \\
{\rm tr} (Q) & < \gamma^2. \label{eqn:clptest4}
\end{align}
Partition $V, V^{-1}$ according to plant and controller state dimensions $n, n_k$ as
\begin{align*}
V = \begin{bmatrix} S & \star \\
N & \star \end{bmatrix}, \qquad
V^{-1} = \begin{bmatrix} R^T & \star \\
M^T & \star \end{bmatrix}, \qquad U := R S + M N, 
\end{align*}
and similarly partition $X, X^{-1}$ as
\begin{align*}
X = \begin{bmatrix}
J & J_2 \\
J_2^T & J_3
\end{bmatrix}, \qquad
X^{-1} = \begin{bmatrix}
L & L_2 \\
L_2^T & L_3
\end{bmatrix}, \qquad L_2 J_2^T = I - L J. 
\end{align*}

Applying the congruence transformation suggested in \cite{DeGB2002}, let 
\begin{align*}
Z := \begin{bmatrix}
I & R^T \\
0 & M^T
\end{bmatrix}. 
\end{align*}
Then
\begin{align*}
V Z = \begin{bmatrix} S & I \\
N & 0 \end{bmatrix}, \qquad Z^T V Z = \begin{bmatrix}
S & I \\
U & R
\end{bmatrix}. 
\end{align*}
Define 
\begin{align*}
W_1 = \begin{bmatrix}
L & I \\
L_2^T & 0
\end{bmatrix}, \qquad
W_2 = \begin{bmatrix}
I & J \\
0 & J_2^T
\end{bmatrix} 
\end{align*}
so that $X W_1 = W_2$ and 
$W_1^T X W_1 = \begin{bmatrix}
L & I \\
I & J
\end{bmatrix}$. 
Then, the transformed closed-loop system satisfies 
\begin{align*}
& \begin{bmatrix} Z^T \left( A_{cl} V \right) Z & Z^T \left(
B_{0,cl} X \right) W_1 & Z^T
B_{1,cl} \\
W_1^T \left( C_{0,cl} V \right) Z & W_1^T \left( D_{00,cl} X \right) W_1 & W_1^T D_{01,cl} \\
\left( C_{1,cl} V \right) Z & \left( D_{10,cl} X \right) W_1 &
D_{11,cl}
\end{bmatrix} \nonumber \\
& = \left[ \begin{array}{cc|cc|c}
A S & A & B_0 & B_0 J & B_1 \\
0 & R A & R B_0 & 0 & R B_1 \\ \hline
0 L C_0 & 0 & L D_{00} & 0 & 0 \\
C_0 S & C_0 & D_{00} & D_{00} J & 0 \\ \hline C_1 S & C_1 & D_{10}
& D_{10} J & 0
\end{array} \right] + \left[ \begin{array}{ccc}
0 & B_2 & 0 \\
I & 0 & 0 \\ \hline
0 & 0 & I \\
0 & D_{02} & 0 \\ \hline
0 & D_{12} & 0
\end{array} \right] \\
& \hspace*{0.25in} \times \begin{bmatrix}
\hat{A}_k & \hat{B}_{k1} & \hat{B}_{k0} \\
\hat{C}_{k1} & 0 & \hat{D}_{k10} \\
\hat{C}_{k0} & 0 & \hat{D}_{k00}
\end{bmatrix}
\left[ \begin{array}{cc|cc|c}
I & 0 & 0 & 0 & 0 \\
0 & C_2 & D_{20} & 0 & D_{21} \\
0 & 0 & 0 & I & 0
\end{array} \right]. 
\end{align*}
The transformed controller matrices relate to the original  
$A_k, B_{k1}, B_{k0}, C_{k1}, C_{k0}, D_{k10}, D_{k00}$ via
\begin{align}
\begin{bmatrix}
\hat{A}_k & \hat{B}_{k1} & \hat{B}_{k0} \\
\hat{C}_{k1} & 0 & \hat{D}_{k10} \\
\hat{C}_{k0} & 0 & \hat{D}_{k00}
\end{bmatrix} & = \begin{bmatrix}
R A S & 0 & R B_0 J \\
0 & 0 & 0 \\
L C_0 S & 0 & L D_{00} J
\end{bmatrix} \nonumber \\
& \hspace*{0.25in} + \begin{bmatrix}
M & R B_2 & 0 \\
0 & I & 0 \\
0 & L D_{02} & L_2
\end{bmatrix}
\begin{bmatrix}
A_k & B_{k1} & B_{k0} \\
C_{k1} & 0 & D_{k10} \\
C_{k0} & 0 & D_{k00}
\end{bmatrix}
\begin{bmatrix}
N & 0 & 0 \\
C_2 S & I & D_{20} J \\
0 & 0 & J_2^T
\end{bmatrix}. \label{eqn:transform}
\end{align}

Multiplying ${\rm diag} \left\{I, I, Z^T, W_1^T, I \right\}$ from the left and its transpose from the right in the 
(\ref{eqn:clptest1}), and defining $T_+ = Z^T P_{+,cl} Z, T_- 
= Z^T P_{-,cl} Z$, we get the condition (\ref{eqn:h2LFTsyn1}). The second condition
(\ref{eqn:h2LFTsyn2}) follows analogously by congruent transformation ${\rm diag} \left\{ I, Z^T \right\}$ to the 
eqn. (\ref{eqn:clptest2}).

Finally, since the matrices
\[
\begin{bmatrix}
M & R B_2 & 0 \\
0 & I & 0 \\
0 & L D_{02} & L_2
\end{bmatrix} \qquad
\begin{bmatrix}
N & 0 & 0 \\
C_2 S & I & D_{20} J \\
0 & 0 & J_2^T
\end{bmatrix}
\]
are nonsingular, the gain-scheduling output-feedback controller can be recovered by inverting eqn. (\ref{eqn:transform}).
\end{proof}

This approach provides a convex synthesis condition for gain-scheduling ${\cal H}_2$ control while reducing conservatism compared to direct synthesis methods.

\section{Illustrative Examples}

In this section, two representative examples are presented to demonstrate the effectiveness of the proposed robust and gain-scheduling ${\cal H}_2$ control framework. 
The first example considers a two-disk mechanical system and illustrates robust state-feedback synthesis. 
The second example involves an active magnetic bearing (AMB) system and demonstrates gain-scheduled output-feedback 
${\cal H}_2$ control for parameter-dependent dynamics.

\subsection{Two-Disk System}

We first consider a two-disk system previously studied in \cite{WuL2004} to illustrate the robust state-feedback synthesis method developed in Section~\ref{sec:robh2syn}. The dynamics of the system are described by 
\begin{align}
M_1 \left[ \ddot{r}_1(t) - \Omega_1^2(t) r_1(t) \right] & = -b
\dot{r}_1(t)
- k ( r_1(t) + r_2(t) ) + f(t) \label{eqn:two-disk1} \\
M_2 \left[ \ddot{r}_2(t) - \Omega_2^2(t) r_2(t) \right] & = -b
\dot{r}_2(t) - k ( r_1(t) + r_2(t) ) \label{eqn:two-disk2}
\end{align}
where $r_1, r_2$ denote the positions of the two sliders relative to the center, and $\Omega_1, \Omega_2$ are the angular velocities of the two rods, varying in the intervals
$[0, 3] rad/sec$ and $[0, 5] rad/sec$, respectively. 
The control input $f$ acts on the first slider. The system parameters are $M_1 = 1.0 kg, M_2 = 0.5 kg$, damping
coefficient $b = 1.0 kg/sec$, and spring constant $k = 200 N/m$.

Define normalized uncertainty parameters 
\[
\delta_1 = \frac{\Omega_1^2}{4.5}-1, \ \delta_2 =
\frac{\Omega_2^2}{12.5}-1 
\]
so that $\delta_1, \delta_2 \in [-1, 1]$.
Let $x_1 = r_1$, $x_2 = r_2$, $x_3 = \dot{r}_1$, $x_4 =
\dot{r}_2$, $u = f$, and $y = r_2$, The plant described by  (\ref{eqn:two-disk1})-(\ref{eqn:two-disk2}) can then be represented in an LFT form as 
\begin{align}
\begin{bmatrix}
\dot{x} \\
q \\
y
\end{bmatrix} & =
\left[ \begin{array}{cccc|cc|cc|c}
0 & 0 & 1 & 0 & 0 & 0 & 0 & 0 & 0 \\
0 & 0 & 0 & 1 & 0& 0 & 0 & 0 & 0 \\
4.5 - \frac{k}{M_1} & -\frac{k}{M_1} & -\frac{b}{M_1} & 0
& 4.5 & 0 & \frac{0.1}{M_1} & 0 & \frac{1}{M_1} \\
-\frac{k}{M_2} & 12.5 - \frac{k}{M_2} & 0 & -\frac{b}{M_2} & 0 &
12.5 & 0 & \frac{0.1}{M_2} & 0 \\ \hline
1 & 0 & 0 & 0 & 0 & 0 & 0 & 0 & 0 \\
0 & 1 & 0 & 0 & 0 & 0 & 0 & 0 & 0 \\ \hline 0 & 1 & 0 & 0 & 0 & 0
& 0 & 0 & 0
\end{array} \right]
\begin{bmatrix}
x \\
p \\
d \\
u
\end{bmatrix} \\
p & = \begin{bmatrix} \delta_1 & \\ & \delta_2 \end{bmatrix} q
\end{align}

All states are assumed to be measurable for state-feedback implementation. The design objectives are specified using the following weighting functions:
\begin{align*}
W_e(s) & = \frac{0.3 s + 1.2}{s + 0.04}, \qquad W_u(s) = \frac{s +
0.1}{0.01 s + 125}, \qquad
W_n(s) = \frac{s + 0.4}{0.01 s + 400}, \\
W_a(s) & = 0.00001, \qquad Act(s) = \frac{1}{0.01 s + 1}.
\end{align*}

The continuous-time plant is discretized using a zero-order hold with a sampling time of $0.01 sec$.
Robust ${\cal H}_2$ and ${\cal H}_\infty$ state-feedback 
controllers are synthesized and compared. 
The resulting induced $\ell_2$-norm performance levels are summarized in Table \ref{tab:perf_sf}.

\begin{table}[htb]
\caption{Performance comparison of robust state-feedback control.}
\label{tab:perf_sf}
\begin{center}
\begin{tabular}{c|c} \hline
Method & Induced ${\ell}_2$ norm \\ \hline
${\cal H}_{\infty}$ state feedback control & 0.898 \\
${\cal H}_2$ state feedback control & 0.478 \\ \hline
\end{tabular}
\end{center}
\end{table}

It is evident that the proposed robust ${\cal H}_2$ synthesis approach achieves substantially less conservative performance than the corresponding ${\cal H}_\infty$ design. 
This improvement can be attributed to the smaller disturbance set inherent in the robust ${\cal H}_2$ formulation, which better reflects stochastic disturbance characteristics.

\subsection{Active Magnetic Bearing System}

The second example considers an active magnetic bearing (AMB) system to demonstrate the proposed gain-scheduling output-feedback ${\cal H}_\infty$ control approach \cite{Wu2001}. 
Due to the linear dependence of the plant dynamics on rotor speed, the nonlinear gyroscopic equations can be approximated by a linear parameter-varying (LPV) model: 
\begin{align}
\ell \ddot{\theta} & = -\frac{\rho J_a}{J_r} \ell \dot{\psi} + \frac{1}{m}
\left(-4 c_2 \ell \theta + 2 c_1 \phi_\theta + f_{d \theta}\right) \label{eqn:ambdyn1} \\
\ell \ddot{\psi} & = \frac{\rho J_a}{J_r} \ell \dot{\theta} +
\frac{1}{m}
\left(-4 c_2 \ell \psi + 2 c_1 \phi_\psi + f_{d \psi}\right) \\
\dot{\phi_\theta} & = \frac{1}{N} \left( e_\theta + 2 d_2 \ell \theta - d_1 \phi_\theta
\right) \\
\dot{\phi_\psi} & = \frac{1}{N} \left( e_\psi + 2 d_2 \ell \psi - d_1 \phi_\psi \right) \label{eqn:ambdyn4}
\end{align}
where the rotor speed $\rho$ serves as the scheduling parameter. 
$\theta$, $\psi$ are the Euler angles denoting the orientation of rotor centerline. 
$J_a$, $J_r$ are the moment of inertia of the rotor in axial and radial directions, respectively. 
$\phi_\theta$, $\phi_\psi$ are the differential magnetic flux from electromagnetic pairs, $e_\theta$, $e_\psi$ are the corresponding differences of electric voltage. 
$f_{d \theta}$, $f_{d \psi}$ are disturbance forces caused by
gravity, modeling errors, imbalances, etc. 
The constants $c_1$, $c_2$, $d_1$, and $d_2$ depend on $\Phi_0$, $G_0$, $R$, $A$, $N$, $\nu_0$, and the bearing geometry, as defined in
\begin{align*}
c_1 & = 2 k \Phi_0 \left( 1 + \frac{2 G_0}{\pi h} \right), \qquad
c_2 = \frac{2 k \Phi_0^2}{\pi h}, \\
d_1 & = \frac{2 R G_0}{\nu_0 A N}, \qquad d_2 = \frac{2 R \Phi_0}{\nu_0 A N}
\end{align*} 
The system parameters are summarized in Table \ref{tab:AMB-para}.
\begin{table}[!htb]
\caption{Active magnetic bearing parameters.}
\label{tab:AMB-para}
\begin{center}
\begin{tabular}{c|c} \hline
Parameter & Value \\ \hline
$A$ area of each pole ($mm^2$) & $1531.79$ \\
$h$ pole width ($mm$) & $40.00$ \\
$G_0$ nominal gap($mm$) & $0.55$ \\
$J_r$ radial moment of inertia ($kg \cdot m^2$) & $0.333$ \\
$J_a$ axial moment of inertia ($kg \cdot m^2$) & $0.0136$ \\
$\ell$ half the length of the shaft ($m$) & $0.13$ \\
$k$ & $4.6755576 \times 10^8$ \\
$N$ number of coil turns & $400$ \\
$R$ coil resistance ($Ohm$) & $14.6$ \\
$\Phi_0$ nominal airgap ($Wb$) & $2.09 \times 10^{-4}$ \\ \hline
\end{tabular} 
\end{center}
\end{table}

Let the state, disturbance, and control vectors as 
\[
x^T = \begin{bmatrix} \ell \theta & \ell \psi & \ell
\dot{\theta} & \ell \dot{\psi} & \phi_\theta &\phi_\psi
\end{bmatrix}, \quad d^T = \begin{bmatrix} f_{d \theta} & f_{d \psi}
\end{bmatrix}, \quad u^T = \begin{bmatrix} e_\theta & e_\psi
\end{bmatrix}, 
\]
the system can be written in LPV state-space form. 
\begin{align*}
\dot{x} = A(\rho) x + B_1 d + B_2 u \\
e = C_1 x + D_{11} d + D_{12} u \\
y = C_2 x + D_{21} d + D_{22} u
\end{align*}
where the state-space data are
\begin{align*}
A(\rho) & = \begin{bmatrix}
0 & 0 & 1 & 0 & 0 & 0 \\
0 & 0 & 0 & 1 & 0 & 0 \\
-\frac{4 c_2}{m} & 0 & 0 & -\frac{\rho J_a}{J_r} & \frac{2 c_1}{m} & 0 \\
0 & -\frac{4 c_2}{m} & \frac{\rho J_a}{J_r} & 0 & 0 & \frac{2 c_1}{m} \\
\frac{2 d_2}{N} & 0 & 0 & 0 & -\frac{d_1}{N} & 0 \\
0 & \frac{2 d_2}{N} & 0 & 0 & 0 & -\frac{d_1}{N}
\end{bmatrix} \\
B_1 & = 0_{6 \times 2}, \qquad B_2 = \frac{1}{N} \begin{bmatrix}
0_{4 \times 2} \\
I_2
\end{bmatrix} \\
C_1 & = \begin{bmatrix}
\begin{array}{cc} I_2 & 0_{2 \times 4} \end{array} \\
0_{2 \times 6}
\end{bmatrix}, \qquad D_{11} = 0_{4 \times 2},
\qquad D_{12} = \begin{bmatrix}
0_{2 \times 2} \\
I_2
\end{bmatrix} \\
C_2 & = \begin{bmatrix}
I_2 & 0_{2 \times 4}
\end{bmatrix}, \qquad D_{21} = I_2, \qquad D_{22} = 0_{2 \times 2}. 
\end{align*}  
The rotor speed $\rho$ is assumed to be available in real-time  and varies within $[315, 1100] rad/s$, a range in which strong gyroscopic effects are present. 
In automatic balancing design, $f_{d \theta}$ and $f_{d \psi}$ are typically modeled as sensor noise on the measured rotor displacement. 
The affine LPV model is subsequently transformed into an LFT representation for gain-scheduling synthesis. 

The control objective is to stabilize the system over the full speed range while minimizing disturbance effects on rotor displacement. 
These objectives are captured using weighting functions derived from frozen-parameter designs, yielding the weighted interconnection shown in Fig.~\ref{fig:wtop}.

\begin{figure}[!htb]
\centering 
\includegraphics[width=3.5in]{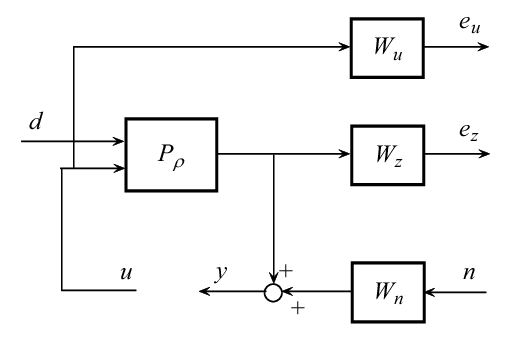} \\
\caption{Weighted open-loop interconnection for the magnetic
bearing system.} 
\label{fig:wtop}
\end{figure}

The weighting functions are chosen as
\begin{align*}
W_z(s) & = \frac{10 (s + 8)}{s + 0.001} I_2, \qquad
W_u(s) = \frac{0.01 (s + 100)}{s + 100000} I_2 \\
W_n(s) & = 0.001 I_2. 
\end{align*}

The continuous-time system is discretized with a sampling time of $0.01 sec$. 
Table~\ref{tab:perf_of} compares the robust 
${\cal H}_2$ and ${\cal H}_\infty$ gain-scheduling designs. 
As in the previous example, the robust ${\cal H}_2$ controller achieves significantly improved performance with reduced conservatism. 
 
\begin{table}[!htb]
\caption{Performance comparison of different gain-scheduling control.}
\label{tab:perf_of}
\begin{center}
\begin{tabular}{c|c} \hline
Method & Induced ${\ell}_2$ norm \\ \hline
${\cal H}_{\infty}$ output feedback control & 8.458 \\
${\cal H}_2$ output feedback control & 2.640 \\ \hline
\end{tabular}
\end{center}
\end{table}

Time-domain simulations are shown in Fig.~\ref{fig:simu}, where step disturbances of magnitude 
$0.001 m$ and opposite signs are applied. 
The gain-scheduled ${\cal H}_2$ controller effectively suppresses disturbances and rapidly drives the rotor displacements to zero across the entire speed range.

\begin{figure}[!htp]
\centering
\begin{minipage}[c]{1\textwidth}
  \begin{minipage}[c]{0.5\textwidth}
  \centering
  \includegraphics[width=2.7in]{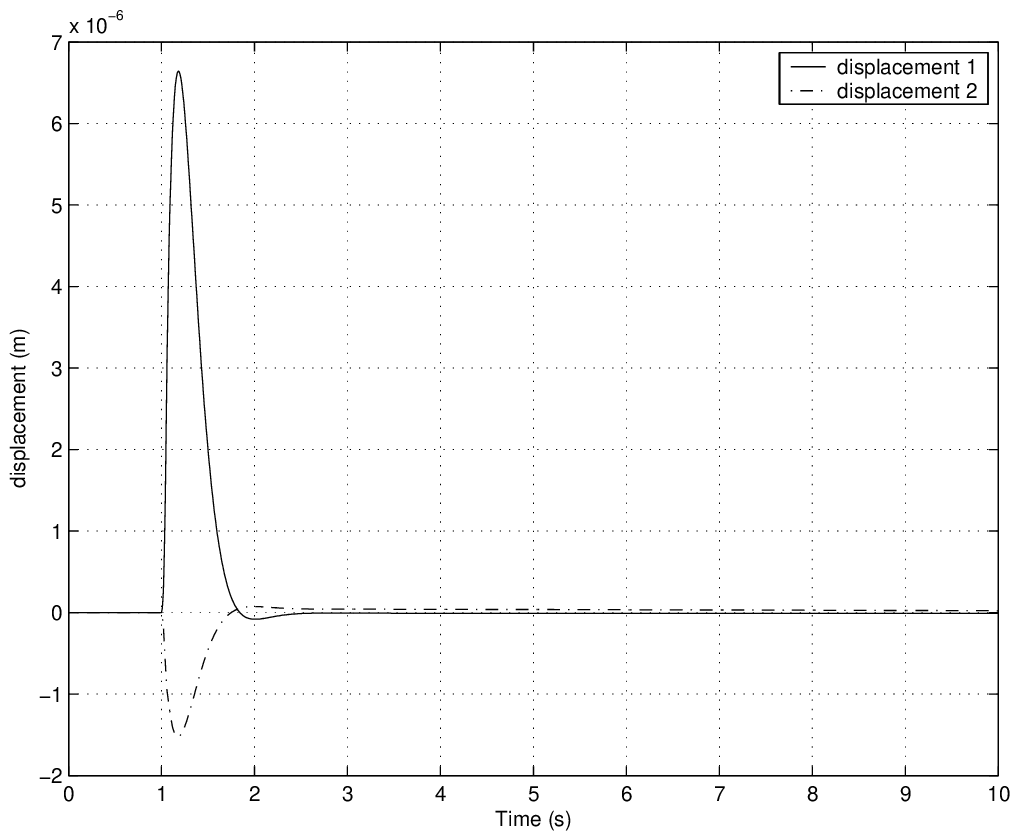} \\
  (a) Displacement
  \end{minipage}
  \begin{minipage}[c]{0.5\textwidth}
  \centering
  \includegraphics[width=2.7in]{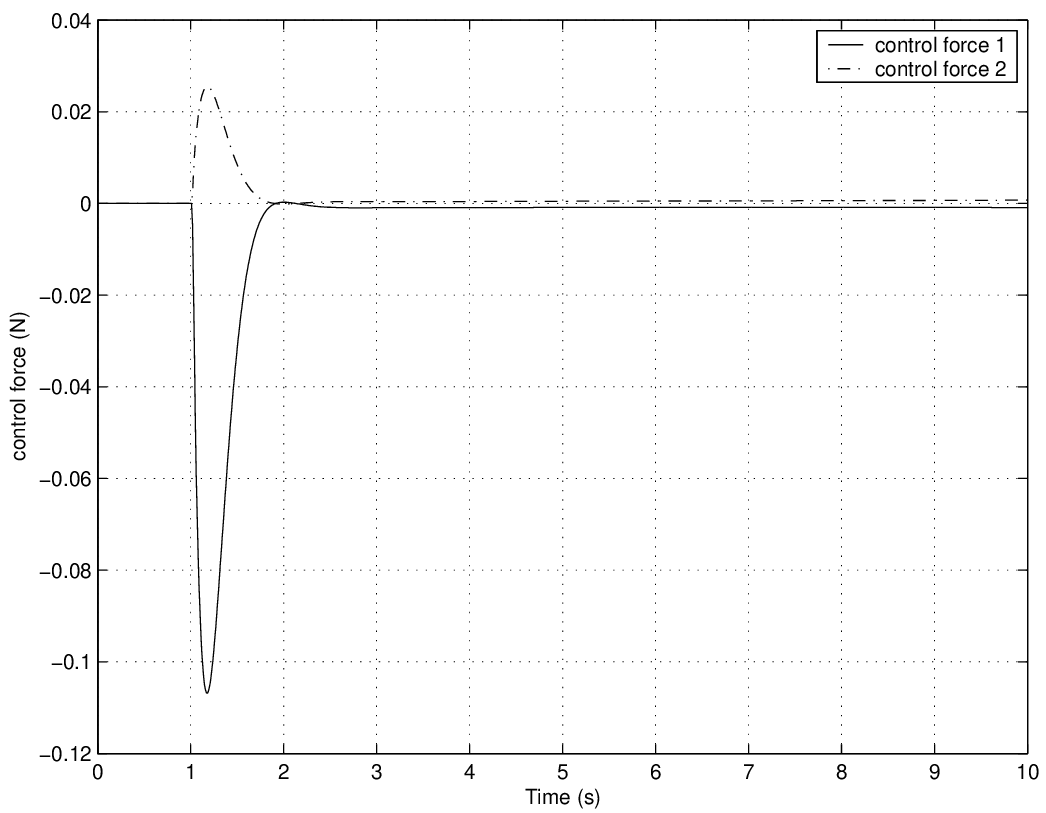} \\
  (b) Control force
  \end{minipage}
\end{minipage}
\caption{Time-domain simulation of gain-scheduling ${\cal H}_2$ control}
\label{fig:simu}
\end{figure}

\section{Conclusion}
\label{sec:conc}

This paper has developed a unified robust and gain-scheduling control framework based on a robust ${\cal H}_2$ performance measure, providing an alternative to classical worst-case design methodologies. 
In contrast to ${\cal H}_\infty$ control, which characterizes performance through induced $\ell_2$-gain and is inherently tailored to adversarial disturbances, the proposed approach is theoretically grounded in the white-noise and impulse-response interpretation of the ${\cal H}_2$ norm. 
By embedding uncertainty descriptions directly into the ${\cal H}_2$ analysis, the resulting performance bounds remain robust while avoiding the excessive conservatism typically associated with worst-case criteria. 
The derived synthesis conditions establish explicit links between robust stability, gain scheduling, and stochastic performance, thereby extending classical ${\cal H}_2$ theory beyond the linear time-invariant setting. 
Numerical examples confirm that the proposed robust ${\cal H}_2$ design achieves improved disturbance rejection with reduced conservatism compared to standard ${\cal H}_\infty$-based controllers, particularly in systems dominated by stochastic disturbances. 

Several directions for future research can be envisioned. First, the extension of the proposed robust 
${\cal H}_2$ synthesis to output-feedback controllers with dynamic compensation could broaden practical applicability, particularly for systems with partial state measurements. Second, integrating performance objectives beyond the 
${\cal H}_2$ norm, such as mixed ${\cal H}_2/{\cal H}_\infty$ 
criteria or regional pole placement, may yield controllers 
that better balance robustness and transient performance.

\end{document}